\documentclass[runningheads,anonymous]{llncs}
\usepackage[T1]{fontenc}
\usepackage{graphicx}
\usepackage{hyperref}
\usepackage{amsmath}
\usepackage{amsfonts}
\usepackage[ruled,linesnumbered, noend]{algorithm2e} 
\usepackage{algpseudocode}
\usepackage{acronym}
\acrodef{bf}[BF]{Bloom filter}
\acrodef{rbf}[RBF]{Rational Bloom filter}
\acrodef{vsbbf}[VSBBF]{Variably-Sized Block Bloom filter}
\usepackage[square,sort,comma,numbers]{natbib}
\usepackage{color}

\begin{document}

\title{Extending the Applicability of Bloom Filters by Relaxing their Parameter Constraints}
\titlerunning{Extending the Applicability of Bloom Filters}
%
\author{Paul Walther
\orcidID{0000-0002-5101-5793} \and
Wejdene Mansour
\orcidID{0009-0008-4362-2092} \and
Johann Maximilian Zollner
\orcidID{0000-0003-3742-8468}\and
Martin Werner
\orcidID{0000-0002-6951-8022}}
\authorrunning{P. Walther et al.}
%
\institute{TUM School of Engineering and Design\\Technical University of Munich, Munich, Germany \\
\email{\{paul.walther,wejdene.mansour,maximilian.zollner,martin.werner\}@tum.de}}
\maketitle              

\begin{abstract}
These days, Key-Value Stores are widely used for scalable data storage. In this environment, \acf{bf}  serve as an efficient probabilistic data structure for representing sets of keys. They allow for set membership queries with no false negatives and with the right choice of the main parameters -- length of the \ac{bf}, number of hash functions used to map an element to the array's indices, and the number of elements inserted -- the false positive rate is optimized. However, the number of hash functions is constrained to integer values, and the length of a \ac{bf} is usually chosen to be a power of two to allow for efficient modulo operations using binary arithmetic. 
In this paper, we relax these constraints by proposing the Rational Bloom filter, which allows for non-integer numbers of hash functions. This results in optimized fraction-of-zero values for a known number of elements to be inserted. Based on this, we construct the Variably-Sized Block \ac{bf} to allow for a flexible filter length, especially for large filters, with efficient computation.
\end{abstract}
This preprint has not undergone peer review or any post-submission improvements or corrections. The Version of Record of this contribution is published in Chrysanthis, P.K., Nørvåg, K., Stefanidis, K., Zhang, Z., Quintarelli, E., Zumpano, E. (eds) New Trends in Database and Information Systems. ADBIS 2025. Communications in Computer and Information Science, vol 2676. Springer, Cham and is available online at \url{https://doi.org/10.1007/978-3-032-05727-3_2}.

\section{Introduction}
Key-Value Stores proved themselves as an efficient storage model to support all steps in the data life cycle \cite{Werner.2015b,DeCandia.2007}. 
They store data as pairs of unique keys and associated values and thereby allow for efficient retrieval and modification of the data. Consequently, an efficient representation of the keys in memory is necessary for the scalability and low-latency handling of Key-Value Stores.
At the same time, large, sparse, low-cardinality data, such as rasterized global building footprints can be modelled as sets \cite{Werner.2021}. As access to this data type is often random, the requirements for an efficient in-memory representation are similar to those of a Key-Value Store. 

For both applications \acp{bf} as probabilistic data structures to store the set property were proposed before \cite{Werner.2021,Lu.2012}. They are configured with three parameters: length $m$ in bits, number of hash functions $k$, and number of elements to store in the \ac{bf} $n$ \cite{Bloom.1970}.
Thereby, the memory footprint trades against the error probability of the data structure which is solely dependent on the three forementioned parameters. However, in practice the optimal choice of parameters is constrained (compare Section \ref{sec:properties}). 

In this context, we propose a method allowing to choose any \textbf{rational number} instead of only an integer number \textbf{of hash functions} $k$, which is very powerful if the set is immutable as it allows for more tailored \acp{bf} and, therefore, decreases false positive rates. 
Further we propose a method to choose \textbf{non-power-of-two sizes for the \ac{bf} $m$ without performance degradation} (by non-uniformity of access) and still without explicit modulo computations. 
The constraints of \acp{bf} and both proposed methods are visualized in Figure~\ref{fig:teaser}. First ideas in this direction were proposed in \cite{Walther.2024}.

\begin{figure}[t]
  \includegraphics[width=\textwidth, trim=0 11.7cm 16.05cm 0, clip]{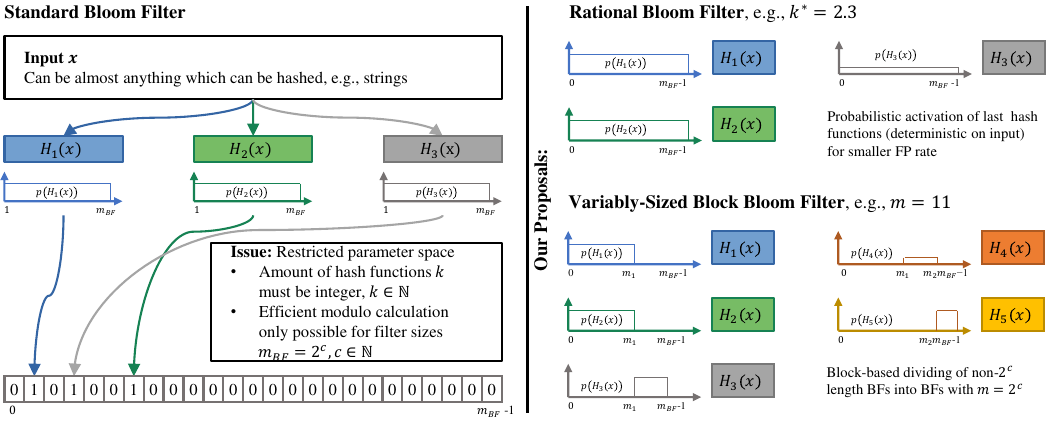}
  \caption{Overview of the proposed approaches to relax \ac{bf} parameter constraints: \acl{rbf} and \acl{vsbbf}}
  \label{fig:teaser}
\end{figure}

\section{Typical Bloom Filter Properties and Constraints}
\label{sec:properties}
The \ac{bf} is a data structure for the efficient probabilistic storage of sets \cite{Bloom.1970}. It consists of a binary array (the filter) and methods for storage in the filter and querying for set membership. The empty filter is an all-zero bit array of length $m$. 
Due to its structure, the \ac{bf} guarantees \texttt{true} for inserted items (no false negatives) but may falsely report uninserted \texttt{1}'s (false positives)~\cite{Bruck.2006,Werner.2021}.

To \textit{insert} an element $x$ into a \ac{bf}, the element is hashed with $k$ uniformly distributed pairwise independent hash functions $H_i(x)$, with $i \in \{1, \dots, k\}$. The hashing maps from the input space of all elements (universe) to the integer range $\{0, \dots,m-1\}$. The \ac{bf} is then set to \texttt{1} at the locations denoted by the same hash functions $H_i(x)$. If the value is already \texttt{1} in one location, it remains \texttt{1} \cite{Gupta.2017}. 
For the \textit{membership query}, the \ac{bf} checks the $k$ indices computed by $H_i(x)$. If any of the denoted values is \texttt{0}, the element is not part of the stored set. Otherwise, if all values are \texttt{1}, the element is present, or it is a false positive error \cite{Gupta.2017}.  
For a given number of elements to store $n$, the \ac{bf} length $m$ in bits, and $k$ hash functions, the false positive rate $p_{\text{FP}}$ of the \ac{bf} can be calculated \cite{Bruck.2006} as
\begin{equation}
    p_{\text{FP}}=\left(1-\left(1-\frac{1}{m}\right)^{kn}\right)^k \approx \left(1 - e^{-kn/m}\right)^k \approx (1-p)^k.
\label{eq:fp_rate}
\end{equation}
For a given $n$, the optimal number of hash functions $k^*$ minimize $p_{\text{FP}}$ with a fraction of zeros $\mathit{foz}=\frac{1}{2}$. This maximizes the entropy of the filter and holds the highest information density~\cite{Werner.2015b}, yielding
\begin{equation}
    k^* = \frac{m}{n}\ln{2}.
    \label{eq:optimalk}
\end{equation}

Nayak and Patgiri explain five main challenges with existing \ac{bf} approaches \cite{Nayak.2019}, namely: 
to reduce the \textit{false positive rate},
the \textit{length adaption} of \acp{bf}, especially with initially unknown dataset sizes,
the \textit{deletion of elements}, which is not possible with a standard \ac{bf} without recalculation of the whole index, 
to implement an \textit{efficient hashing} method, which does not negatively influence the performance of the \ac{bf} and 
the correct determination of the \textit{number of hash functions} to use. 
An approach to these challenges is the relaxation of the constraints for the filter length $m$, the number of hash functions $k$, and the number of elements to store in a \ac{bf} $n$.
With that, an unconstrained, optimal choice of the number of hash functions leads to improved false positive rates (compare Eqn. \ref{eq:fp_rate}), and new approaches to variable filter lengths might improve the scalability of the approach and allow for more efficient hashing and more equally distributed false positive rates~\cite{Estebanez.2014,Almeida.2007}.

In the known literature, the \textit{number of hash functions} $k$ is restricted to integers, as applying, e.g., $0.3$ hash functions is not defined. This reduces flexibility in constructing the optimal filter for given $m$ and $n$. Especially for small $k$, this may result in a non-optimal false positive rate. 
Figure~\ref{fig:amounthashfunctions} visualizes this effect.
\begin{figure}[t]
    \centering
    \includegraphics[width=\textwidth]{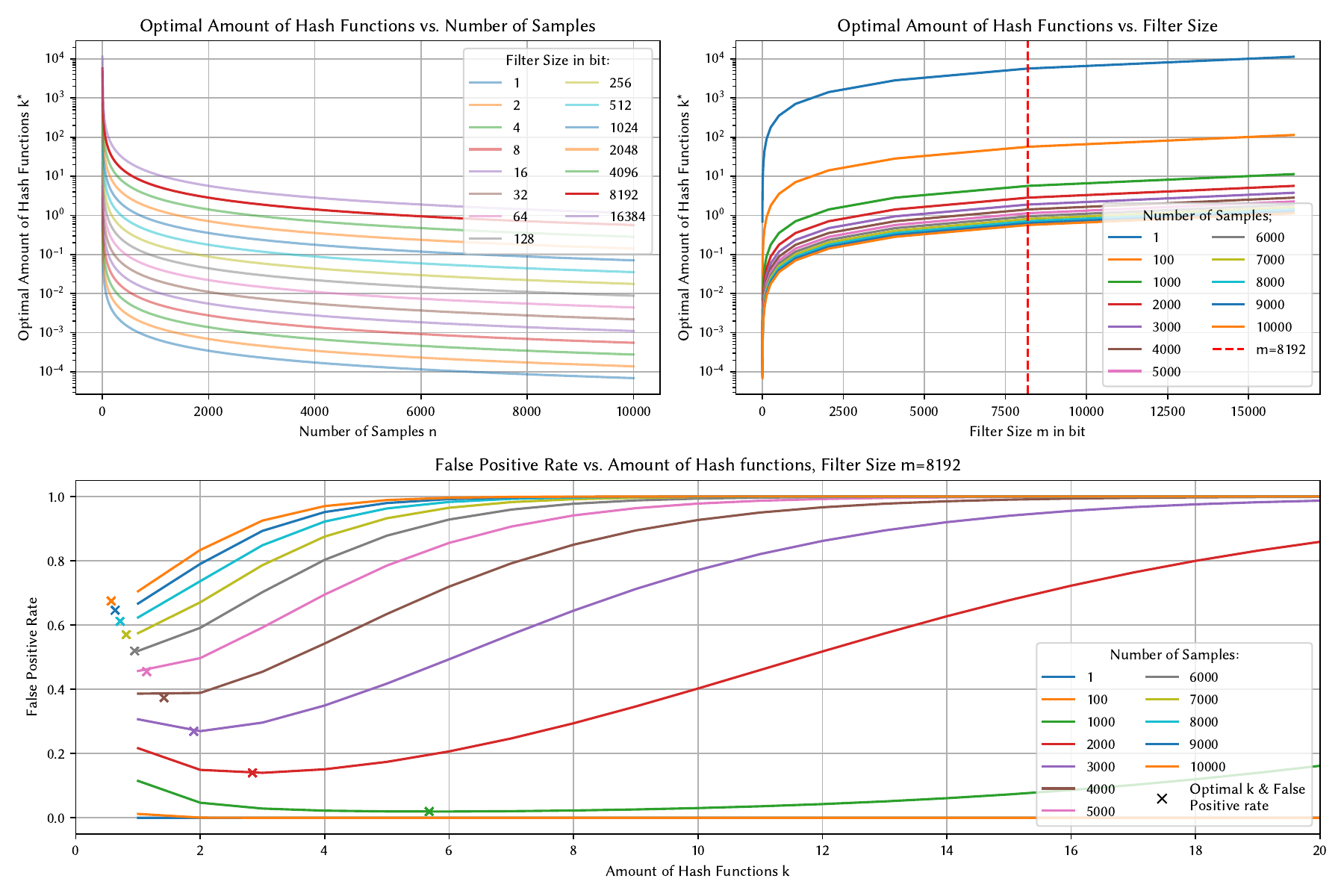}
    \caption{Dependencies between the optimal number of hash functions $k^*$ and number of samples $n$ (top left) and filter size $m$ (top right). Dependency of the false positive rate for a fixed filter size of $m=8192$ on varying numbers of hash functions $k$ and number of samples (bottom). The calculated optimal configuration for each number of samples is marked with a cross.}
    \label{fig:amounthashfunctions}
\end{figure}
Furthermore, the false positive probability is unevenly distributed across elements in multi-hash \acp{bf} \cite{Almeida.2007}, as hash collisions, that means similiar slots denoted by different hash functions for one element, increase false positives. To avoid this effect using less hash functions is beneficial. 


The \textit{length of the \ac{bf}} $m$ is usually chosen to be a power-of-two, $m=2^c, c \in \mathbb{N}_0$ due to efficiency considerations: General hash functions $h_i(x)$ have to be modified to specific hash functions $H_i(x)$ mapping the input space to indices of the \ac{bf} $\{0, \dots,m-1\}$. The simplest method is taking the modulo $h_i(x)\%m$. It is most efficient for power-of-two filter lengths as it simplifies to a binary AND ($\wedge$) operation~\cite{Estebanez.2014}: 
\begin{equation}
    H_i(x)=h_i(x)\mod{m} =h_i(x) \wedge (2^c -1)
    \label{eq:modulo}
\end{equation}
Consequently, when high performance is desired, the selectable lengths for the \ac{bf} always double from one option to the next. This might not be a problem for small \acp{bf}, but for large \acp{bf}, the step sizes become huge. 
For example, for a \ac{bf} size of 16 GB, the neighboring values 8 GB and 32 GB are already very far from each other. Especially in computer systems, which often have memory in the same size intervals, this may become inefficient: Some space is always used, e.g., by the operating system. Therefore, the \ac{bf} may have at most half of the available memory size. 


\section{Rational Bloom Filter}
\label{sec:methods:sub:rationalbloom}
For \acp{bf}, the theoretically optimal number of hash functions, denoted by Eqn. \ref{eq:optimalk}, is generally not an integer but a rational number. An example would be a \ac{bf} of size $m=10$ and $n=5$ elements to store, resulting in an optimal number of $\approx 1.386$ hash functions. This is not feasible in traditional \acp{bf} as they only allow for an integer number of hash functions and, therefore, require an approximation of the optimal number of hash functions with the next integer. However, allowing a non-integer number of hash functions poses advantages: For an exactly known number of elements to store and a \ac{bf} length constrained, e.g., by hardware, the optimal number of hash functions improves the false positive rate. Further, a rational number of hash functions allows for more flexibility in selecting other \ac{bf} parameters. 
Practically, this raises the question of how to apply a rational number of hash functions. For this, we propose a probabilistic approach.
\begin{definition}
    \label{def:probhashfunc}
    A \textbf{probabilistically activated hash function} is a hash function that is not applied to every sample.
    Instead, it is only activated with a probability of activation $0\le p_{activation}\le1, p_{activation} \in \mathbb{R}$.
\end{definition}
In \acp{bf} the activation probability is the rational part of the number of hash functions: $p_{activation}=k-\lfloor k\rfloor$. While this initially seems to add extra effort, only small additional costs are implied in practice for $k>1$. The always-applied hash functions $H_r(x)$, with $r \in [1,\lfloor k \rfloor]$, give enough pseudo-random information that can be used to decide whether the probabilistically activated hash functions should be used for a specific sample. For a particular given sample, this activation is additionally deterministic. 
Based on this, we develop the \ac{rbf} (visualized in Figure \ref{fig:rationalbloom}).
\begin{definition}
    A \textbf{\acf{rbf}} is a \ac{bf} with a non-integer number $k\in \mathbb{R}^+$ of hash functions, consisting of standard \ac{bf} procedures for $\lfloor k\rfloor$ hash functions, while a probabilistically activated hash function represents the non-integer part $k-\lfloor k\rfloor$ (compare Definition \ref{def:probhashfunc}). 
\end{definition}
\begin{figure}
    \centering
    \includegraphics[width=0.5\textwidth, trim=0 16.2cm 25.345cm 0, clip ]{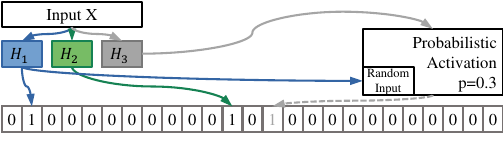}
    \caption{The \ac{rbf} for $k=2.3$, with probabilistically activated hash function $H_3(x)$.}
    \label{fig:rationalbloom}
\end{figure}
In practice this can be implemented by Algorithm \ref{alg:rationalbloom}:
A hash function maps uniformly from the input domain to the integer set $\{0, \dots m\}$. Comparing the obtained hash value to the proportion of the maximum value $p \cdot m$ allows for overall probabilistic activation of $H_{\lfloor k\rfloor+1}$, which is still deterministic with regards to the input element $x$ (compare Algorithm \ref{alg:rationalbloom}, Line \ref{alg:rationalbloom:decision}).
Furthermore, the \ac{rbf} proposal actually allows for the hashing trick by Kirsch and Mitzenmacher \cite{Kirsch.2006}, which gives an additional sufficient independent hash function by only an additional cheap multiplication and addition.
\begin{algorithm}
\footnotesize
\caption{Application of Hash Functions in the \ac{rbf}}
\label{alg:rationalbloom}
\KwData{Element $x$, \ac{bf} $BF$ with length $m$, set of always-applied hash functions $\{H_1, \dots, H_{\lfloor k \rfloor}\}$, probabilistically activated hash function $H_{\lfloor k\rfloor+1}$, rational number of hash functions $k$\;}
\KwResult{Set \ac{bf} bits for indices denoted by the rational number of hash values of input element x}
\For{each $H_j$ in $\{H_1, \dots, H_{\lfloor k \rfloor}\}$}{
    Set $BF[H_j(x)] \gets 1$\;
}
Set $p_{activation} = k - \lfloor k \rfloor$\;
Random hash value $H_r(x) \in [0, m]$;  // No additional calculation needed if we, e.g., choose $H_r(x)=H_{\lfloor k\rfloor}(x)$\;
\If{$H_r(x) < (p_\text{activation}\cdot m)$}{\label{alg:rationalbloom:decision}
    Set $BF[H_{\lfloor k\rfloor+1}(x)] \gets 1$\;
}
\end{algorithm}
\begin{lemma}
    \label{lemma:rational_no_fn}
    The \ac{rbf} has no false negatives.
\end{lemma}
\begin{proof}
    The given \ac{rbf} with $k$ normal hash functions and one probabilistically activated hash function $H_{\lfloor k\rfloor+1}$ has at least as many \texttt{1}-bits in the filter as the normal \ac{bf} $BF$ with $\lfloor k\rfloor$  hash functions
    as the probabilistically additionally activated hash function $H_{\lfloor k\rfloor+1}(x)$ can only set additional \ac{bf} slots to one. 
    A false negative would occur if and only if, during a membership query for an element in the set, the hash values would expect a \texttt{1} in the \ac{bf} where there is no actual \texttt{1}. For all hash functions $H_0$ to $H_{\lfloor k\rfloor}$, this cannot be the case by definition of the standard \ac{bf}. For the probabilistically activated hash function $H_{\lfloor k\rfloor+1}$, this would only be true if it was not activated during insertion but is activated during the query. 
    As the activation is deterministic with respect to the input -- due to the creation through a hash of the input -- an activation in the query without activation during insertion is impossible. Therefore, we can state that the \ac{rbf} has no false negatives.
\end{proof}
\begin{theorem}
    The false positive rate $p_{\text{FP}}^{\text{RBF}}$ of the \ac{rbf} is smaller or equal to the false positive rate of a normal \ac{bf} $p_{\text{FP}}^{\text{BF}}$: $p_{\text{FP}}^{\text{RBF}} \leq p_{\text{FP}}^{\text{BF}}$
\end{theorem}
\begin{proof}
    As described in Eqn. \ref{eq:optimalk}, the optimal number of hash functions $k^*$ is solely determined by the number of entries $n$ and the length of the \ac{bf} $m$ with the assumption that the highest information can be stored for a fraction of zero $\mathit{foz}_{opt}=\frac{1}{2}$. For one standard \ac{bf} with $k_{BF}$ hash functions and one \ac{rbf}, we further assume that we choose the number of hash functions of the \ac{rbf} $k_{\text{RBF}}$ to be optimal, $k_{\text{RBF}}=k^*$.  Then we can differentiate three cases: \\
    Case a) $k_{\text{BF}}>k_{\text{RBF}}$: For the normal \ac{bf}, there are more than the optimal $k^*$ hash functions chosen. Therefore, the $\mathit{foz}$ is smaller than in the optimal case and the false positive rate will increase as it is more likely to accidentally hit a \texttt{1} in the filter $\implies p_{\text{FP}}^{\text{RBF}}<p_{\text{FP}}^{\text{BF}}$.\\
    Case b) $k_{\text{BF}}<k_{\text{RBF}}$: For a smaller than optimal number of hash functions, the filter will become less full. At the same time, each sample is hashed by less-than-optimal hash functions. In this case, with the exact formula for the false positive rate $p_{\text{FP}}=\left(1-\left(1-\frac{1}{m}\right)^{kn}\right)^{k}$$\implies p_{\text{FP}}^{\text{RBF}} \le p_{\text{FP}}^{\text{BF}}$. \\
    Case c) $k_{\text{BF}}=k_{\text{RBF}}$: In this case, the optimal $k^*$ is an integer. Therefore, $BF=RBF$ and consequently $\implies p_{\text{FP}}^{\text{RBF}} = p_{\text{FP}}^{\text{BF}}$.
\end{proof}

\section{Variably-Sized Block Bloom Filters}
\label{sec:methods:sub:blockedbloom}

To enable the efficient computation of hash functions for $m \neq 2^c, c \in \mathbb{N}_0$, we propose the Variably-Sized \acp{bf} (VSBF). A VSBF still makes use of the easy computation of modulo with the binary \& by restricting the hash values only to subsets of the whole \ac{bf} slots. This allows for several advantages compared to standard \acp{bf}: the corresponding \ac{bf} can easily be adapted to a variable size without compromising hashing speed, sub-parts of the filter can directly be used as compressed versions of the stored set, and the false positives due to overlapping hash values are more evenly distributed between all elements. 

\begin{definition}
    A \textbf{Variably-Sized Bloom Filter} (VSBF) is a \ac{bf} of length $m_{\text{BF}}$, with $2^c<m_{\text{BF}}<2^{c+1}, c\in \mathbb{N}_0$. 
    \label{def:variable_sized_bf}
\end{definition}

Consequently, the VSBFs filter length $m_{\text{BF}}$ is unequal to a power-of-two. If we number each slot of the VSBF from $0$ to $m_{\text{BF}}-1$ we can also represent this set of indices $I_{\text{BF}}$ with $J$ non-overlapping subsets $I^j_{\text{BF}}$ each having a length $m_j, j\in \{1, \dots, J\}$. We thereby restrict the $m_j$ to be a power of two.
\begin{lemma}
    A random choice of $m_{\text{BF}}$ can be decomposed into summands $m_j$ for which $m_j =2^c$ with $c\in \mathbb{N}$ such that $m_{\text{BF}}=\sum_j m_j$, with $m_{\text{BF}}, m_j \in \mathbb{N},  j\in \{1, \dots, J\}$.
\end{lemma}
\begin{proof}
    This is directly based on the fact that every integer can be represented in binary. A visualization of an example for $m_{\text{BF}}=11$ is given in Figure \ref{fig:variablesizedbloom}.
\end{proof}
\begin{figure}
    \centering
    \includegraphics[width=0.5\textwidth, trim=0 17.65cm 25.345cm 0, clip]{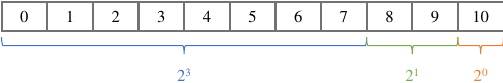}
    \caption{\ac{bf} with indexed slots $I_{\text{BF}}=\{0, ..., 10\}$ and how three subsets, each of length, can represent it, $2^{c_j}$: $I_{\text{BF}}^1=\{0, 1, 2, 3, 4, 5, 6, 7\}$, $I_{\text{BF}}^2=\{8, 9\}$, $I_{\text{BF}}^3=\{10\}$.}
    \label{fig:variablesizedbloom}
\end{figure}
\begin{definition}
    Given a \ac{bf} being decomposed into sets of indexed slots $I^j_{\text{BF}}$, we can define hash functions $H_j(x)$, which only map onto their respective subset $I^j_{\text{BF}}$ each. We call those \textbf{subset hash functions}. 
\end{definition}
These subset hash functions can be calculated efficiently with modulo and addition operations if the subset lengths $m_j$ in the VSBF are defined to be powers-of-two: $H_j(x)=h(x)\%m_j+\min\left(I_{\text{BF}}^j\right)$.
\begin{definition}
    The \textbf{\acf{vsbbf}} is a specification of the previously defined VSBF (compare Definition \ref{def:variable_sized_bf}) where the actual filter of length $m_{\text{BF}}$ is subdivided into $J$ blocks of sizes $m_j, j\in \{1,\dots,J\}$ with $m_j>m_{j+1}$ and $m_j=2^c, c\in\mathbb{N}$. Each block is denoted by a set of indices $I^j$ and is filled by $k_j$ subset hash functions $H_j^i(x)$, which only map to the corresponding indices $I^j$. $k_j$ is thereby the optimal number of hash functions for the respective filter block and might be rational.
\end{definition}
The idea behind this \ac{vsbbf} is visualized in Figure \ref{fig:blockedbloom}. There the \ac{bf} of length $m_{\text{BF}}=25$ is subdivided into three subsets of length $m_1=2^4$, $ m_2=2^3$, and $m_3=2^0$.
\begin{figure}
    \centering
    \includegraphics[width=0.5\columnwidth, trim=0 14.26cm 25.345cm 0, clip]{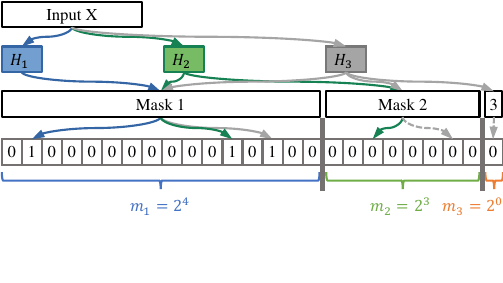}
    \caption{Visualization of a \ac{vsbbf} of length $m_{\text{BF}}=25$ with $m_1=2^4=16$, $m_2=2^3=8$ and $m_3=2^0=1$.}
    \label{fig:blockedbloom}
\end{figure}

The implementation of the \ac{vsbbf} consists of three parts: The calculation of optimal BF blocks including the corresponding number of hash functions~$k_j^*$, the insertion of elements in the proposed data structure, and the query of set property for given elements from the proposed data structure. For the calculation of the Block sizes $m_j$ a simple restructuring of the binary representation of the overall \ac{bf} size $m_{\text{BF}}$ is possible: With this representation being denoted as~$(m_{\text{BF}})_2$ with~$m_i, i\in \{\lfloor log_2(m_{\text{BF}})\rfloor,...,0\}$ denoting the binary digits, the desired \ac{bf} sizes are $\{2^{m_i}\} \forall i| m_i=1$. The corresponding number of hash functions for each block is then denoted as~$k_j^*=\frac{m_j}{n}\ln{2}$ and can be calculated dynamically.

Compared to the standard with one long \ac{bf}, the Block \ac{bf} needs to map the hash values to various filter lengths $m_i$.
This can be efficiently done, as all subparts of the filter are still of size $2^c, c\in \mathbb{N}$, and thus we can apply the efficient modulo operation described in Eqn.~\ref{eq:modulo}.
For a non-integer optimal number of hash functions $k_j$, we apply the principles of \acp{rbf} as described in Section~\ref{sec:methods:sub:rationalbloom}. The whole algorithm for insertion is given in Algorithm~\ref{alg:blockbloom_insert}. For simplicity, this algorithm does not consider rational numbers of hash functions. Querying an element is similar to insertion. Once the hash values are obtained, checking for \texttt{1}'s works similarly to a standard BF. 
\begin{algorithm}
\footnotesize
\caption{Inserting an Element in the \ac{vsbbf}}
\label{alg:blockbloom_insert}
\KwData{Element $x$, Total elements to insert $n$, \ac{vsbbf} $BF$ of total length $m_{\text{BF}}$, Two uniform, pairwise independent hash functions $h_1(x), h_2(x)$}
\KwResult{Inserted element $x$ into the \ac{vsbbf}}
binary $\gets$ BinaryRepresentation($m_{\text{BF}}$)\;
length $\gets$ Length(binary)\;
hash\_values $\gets h_1(x), h_2(x)$\;
offset $\gets 0$\;
\For{j $\gets$ 0 \KwTo length - 1}{
    \If{binary[$j$] == 1}{
        $m_j \gets 2^{\text{length} - j - 1}$\;
        $k_j \gets \frac{m_j}{n} \cdot \ln{2}$\;
        \For{$i \gets 0$ \KwTo $k_j$}{
            $H_i(x) = \left( (hash\_values_1 + (i + m_j) \cdot hash\_values_2) \& (m_j - 1) \right) + \text{offset}$\;
            $BF[H_i(x)] \gets 1$\;
            offset += $m_j$\;
        }
    }
}
\end{algorithm}

\begin{lemma}
    The \ac{vsbbf} has no false negatives.
\end{lemma}
\begin{proof}
    As using a \ac{vsbbf} is equal to using several normal \acp{bf} of various sizes for the same set, the proof boils down to none of the filter blocks allowing for false negatives. 
    By definition, the filter blocks are all either normal or \acp{rbf}, and thus, it follows from Lemma \ref{lemma:rational_no_fn} that the Block \ac{bf} has no false negatives.
\end{proof}

The false positive rate of the \ac{vsbbf} can then be calculated with the chain rule:
\begin{theorem}
    The false positive rate of the \ac{vsbbf} is:
    \begin{equation}
        p_{\text{FP}}^{\text{Block}} =\prod_j p_{\text{FP}}^j\approx \prod_j \left(1-e^{-k_jn/m_j}\right)^{k_j}
    \end{equation}
    Where $k_j$ denotes the number of hash functions and $m_j$ is the length of the $j$-th subset \ac{bf}. 
\end{theorem}
\begin{proof}
    In general, there are four classes of results for a set query: a true positive or true negative answer denotes the right functioning of the \ac{bf}. False negatives are not possible by construction (compare Section \ref{sec:properties}), and a false positive happens if a sample is described as being part of the set which was actually not inserted in the \ac{bf}. 
    For several \acp{bf} representing the same set, which is similar to hash functions only mapping to subsets of the full \ac{bf}, deciding whether a queried sample is in the desired set is always a combination of the comparison of values for all applied hash functions. An item can only be wrongly denoted to be in the set if all hash functions for all filters point wrongly to a \texttt{1} element. This is equal to all subset \acp{bf} wrongly denoting the element to be in the set. 
\end{proof}

\begin{corollary}
    For the right choice of $k_j$, the false positive rate of the combined filters of length $m_j=2^{c_j}$ stays the same as the single \ac{bf} of size $m_{\text{BF}}$. This means:
    \begin{equation}
        p_{\text{FP}}^{\text{Block}}\approx \prod_j \left(1-e^{-k_jn/m_j}\right)^{k_j}=\left(1-e^{-kn/m_{\text{BF}}}\right)^k \approx p_{\text{FP}}^{\text{BF}}.
    \end{equation}
\end{corollary}
\begin{proof}
    With Eqn. \ref{eq:optimalk} for optimal $k^*_j=\frac{m_j}{n} \ln{2}$ and  $k^*=\frac{m_{\text{BF}}}{n} \ln{2}$:
    \begin{equation}
        \prod_j \left(1-e^{-\frac{m_j}{n}\ln 2 \frac{n}{m_j}}\right)^{\frac{m_j}{n}\ln 2}=\left(1-e^{-\frac{m_{\text{BF}}}{n}\ln 2 \frac{n}{m_{\text{BF}}}}\right)^{\frac{m_{\text{BF}}}{n}\ln 2}
    \end{equation}
    \begin{equation}
        \prod_j \left(1-e^{-\ln 2}\right)^{\frac{m_j}{n}\ln 2}=\left(1-e^{-\ln 2 }\right)^{\frac{m_{\text{BF}}}{n}\ln 2}\\
    \end{equation}
    and as $\sum_j m_j=m_{\text{BF}}$ by construction:
    \begin{equation}
        p_{\text{FP}}^{\text{Block}} = p_{\text{FP}}^{\text{BF}}
    \end{equation}
    for the choice of optimal $k_j$. 
\end{proof}
Here, the previous proposition of \acp{rbf} comes in handy, as a non-optimal choice of the number of hash functions due to integer rounding would introduce additional errors.

\begin{theorem}
    The Block \ac{bf} improves the equal distribution of false positives over the to-be-tested elements compared to a standard \ac{bf} of the same size. 
\end{theorem}
\begin{proof}
    In general, an element has a higher probability of being false positive if its footprint contains fewer \texttt{1}'s, as only the non-zero bits hold information, and it is more probable to have overlaps with an already inserted one then. We denote that different hash functions may produce equal hash indices for one input element as a \textit{clash}. 
    For a fully filled filter with $\mathit{foz}=0.5$, a new item $x_{\text{FP}}$ is denoted false positive if, for all $k$ hash functions, the value is one, although the item was originally not inserted in the filter. The false positive rate is then:
    \begin{equation}
        p_{\text{FP}}=0.5^k.
    \end{equation}
    If hash values of $o$ hash functions clash, the probability of a false positive is increased:
    \begin{equation}
        p_{\text{FP}}^{\text{clash}}=0.5^{(k-o)} > p_{\text{FP}}
    \end{equation}
    Thus, it is desirable to avoid overlaps in hash values of different hash functions for one element. 
    
    For a given Block \ac{bf} with blocks of length $m_j$ the optimal number of hash functions for each block is $k^*_j=\frac{m_j}{n} \ln{2}$ (compare Eqn. \ref{eq:optimalk}). So the sum of applied hash functions is equal to the optimal number $k^*$ for a filter of the summed lengths of the blocks $m=\sum_j m_j$:   
    \begin{equation}
        k_{\text{sum}}=\sum_j k^*_j=\sum_j \frac{m_j}{n} \ln{2}= \frac{\ln{2}}{n} \sum_j m_j = \frac{m}{n} \ln{2} \stackrel{\text{Eq.} \ref{eq:optimalk}}{=} k^*
    \end{equation}
    The probability for one element $x$ being inserted into a filter $BF$ with $k^*$ hash functions having less than $k^*$ hash values is equal to the probability of hash functions delivering the same hash values. For two uniformly distributed hash functions, the probability that they have the same value is $\frac{1}{m}$ if $m$ is the length of the \ac{bf}.
    For three hash functions, it is $\frac{1}{m}+\frac{2}{m}$. So, for k hash functions, it is 
    \begin{equation}
        p_{\text{clash}}=\sum_{i=1}^{k-1} \frac{i}{m}
        \label{eq:clash}
    \end{equation}
    In comparison, for a Block \ac{bf}, the clash probability is:
    \begin{equation}
    \begin{aligned} 
        p_{\text{clash}}^{\text{Block}}&=\sum_j \sum_{i=1}^{k_j-1} \frac{i}{m_j}
    \end{aligned}
    \label{eq:clash_blocked}
    \end{equation}
    With mathematic reformulation and the known properties, we can then show that for all $J>1$:
    \begin{equation}
    \begin{aligned}
        p_{\text{clash}} &> p_{\text{clash}}^{\text{Block}}.
    \end{aligned}
    \end{equation}
\end{proof}

\begin{corollary}
    The proposed solution requires no storage overhead for the description of the block sizes $m_j$ and the number of hash functions $k_j$.
\end{corollary}
\begin{proof}
    Additional information is the block sizes $m_j$ and the number of hash functions being used in each block $k_j$. 
    This information can be easily obtained from the information of the non-blocked representation and two implicit rules applied during construction. First, blocks are sorted from large to small with $m_{j+1}<m_j$. Second, the optimal number of hash functions $k_j^*$ is always chosen for each block, and if necessary, a \ac{rbf} (Section \ref{sec:methods:sub:rationalbloom}) is applied for the blocks to keep a controlled false positive rate. Hence, the block sizes $m_j$ are directly encoded in the binary representation of the whole \ac{bf} size $m_{\text{BF}}$, and the corresponding number of hash functions can be calculated with: 
    \begin{equation}
        k_j^*=\frac{m_j}{n}\ln{2}.
    \end{equation}

\end{proof}
\begin{conjecture}
    The proposed solution can save on processing time $t$, which means $t(BF)\geq t(BF^{\text{Block}})$ for insertion as well as querying for filter sizes $m_{\text{BF}}\neq 2^c, c \in \mathbb{N}_{0}$ and $\sum_j m_j^{\text{Block}}=m_{\text{BF}}$. 
\end{conjecture}
This is true if the runtime of one modulo calculation by $m_{\text{BF}}\neq 2^c$ is computationally more expensive than the computation of $J$ modulo operations with the binary bit trick (compare Eqn.~\ref{eq:modulo}).
$J$ thereby denotes the number of blocks in the corresponding Block \ac{bf} and can be computed by the hamming weight of the binary representation of the full filter length $(m_{\text{BF}})_2$ \cite{Reed.1954}.

A further advantage of this Block \ac{bf} is the easy and meaningful compression by simply taking subsets of the filter, which comprise a certain number of blocks. In tendency, this allows for more compression steps compared to the standard \ac{bf} of size $m_{\text{BF}}=2^c$ as not only halving the size is possible but also any combination of block sizes applied in the \ac{bf}. Available compression ratios are $\frac{\sum_i m_i}{m_{\text{BF}}}, \text{ where } i  \subseteq j$.

\section{Implementation Artifacts of the Proposed Approaches}
To verify our theoretical considerations, we implemented the \ac{rbf} and \ac{vsbbf} in C++ with a Murmur Hash using the hashing trick and efficient modulo calculation with binary arithmetic. As a baseline, we also implemented a standard \ac{bf} with the same efficiency measures.

Interestingly, implementing the \textit{\ac{rbf}} showed artifacts that were not expected and need further investigation. The false positive rate has local minima not at the calculated optimal rational value of hash functions, but especially at nearby integer values, as shown in Figure \ref{fig:rational_bloom_experiment}. For example, the false positive rate for a filter with $m=131,072$ and $n=60,000$ (brown line in the bottom graph) has its global minimum at $k=2$ and another local minimum at $k=1$, although the optimal value lies in between (brown cross).
\begin{figure}[t]
    \centering
    \includegraphics[width=\linewidth]{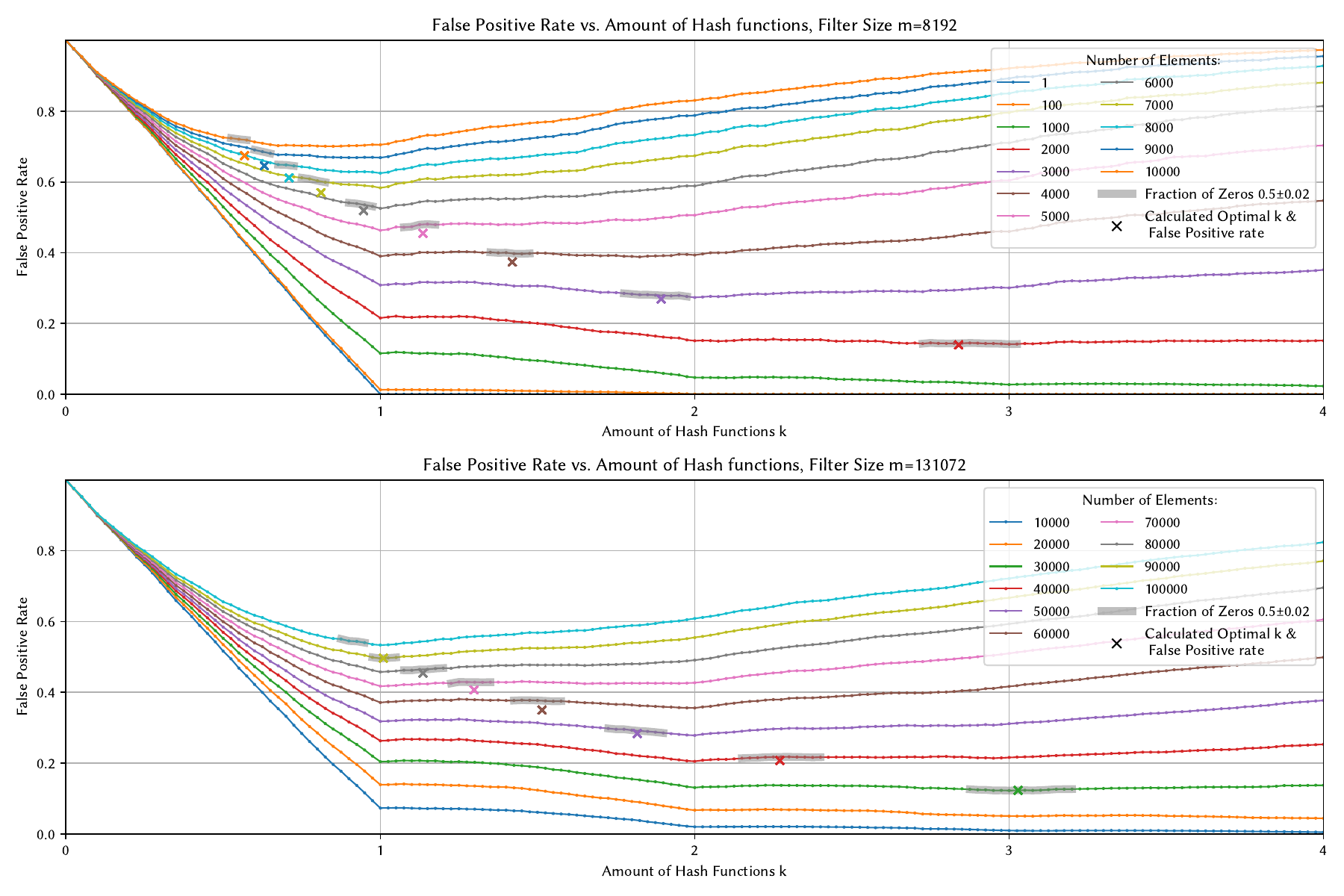}
    \caption{False positive rate of a \acp{rbf} with filter sizes 8,192 (top) and 131,072 bit (bottom) for various numbers of hash functions $k$ and number of elements inserted $n$. To obtain the false positive rate, the filter was queried for 10,000 elements that had not been inserted. Experiments with a fraction of zero of 0.5±0.02 are shadowed grey. Theoretically, calculated optimal configurations are denoted by a cross.}
    \label{fig:rational_bloom_experiment}
\end{figure}
This result is counterintuitive to the authors, and we cannot provide a satisfactory explanation for this issue. We believe this is likely due to the non-uniform behavior of the hashing algorithm. However, similar artifacts were also observed with other hash functions like SHA256, which should lead to a more uniform distribution of hash values. Additionally, larger filter sizes and numbers of elements inserted did not improve this behavior, although according to logical reasoning, this should reduce the effect of non-uniform hash functions. 
When comparing the fraction of zero $\mathit{foz}$ of the close-to-optimal \ac{rbf} to that of the \acp{bf} with an integer number of hash functions, we observe that the \ac{rbf} actually achieves better $\mathit{foz}$ rates. 
For the \textit{\ac{vsbbf}}, we achieve similar rates of false positives compared to a standard \ac{bf}. Improvements in insertion time could not be measured.

Although our practical experiments reveal that proofed theoretical gains could not be fully materialized in the real filters, we still think that the proposed approaches serve as a valuable contribution to \ac{bf} theory, as they allow new possibilities for applications. We see special benefits in future explorations of learned \ac{bf} approaches and filters with hash functions depending on inputs.

\section{Related Work}
To the best of the authors' knowledge, there were no previous approaches to rational numbers of hash functions $k$, while many methods improve the \textit{hashing methods}. 
The most direct way to hash for a \ac{bf} application are \textit{perfect hashing} schemes, where $k$ hash functions map directly from the input space to exactly $m$ hash buckets uniformly \cite{Gupta.2017}. In practice, most \ac{bf} approaches use popular hash functions, like Murmur Hash \cite{Appleby.2008}, with consecutive modulo operations to do so. 
To reduce computation for $k>2$ hash functions, techniques like \textit{double hashing} are used \cite{Gupta.2017}. Instead of calculating $k$ independent hashes, only two are calculated and combined like $h_i(x)=h_1(x)+f(i)h_2(x)$ \cite{Kirsch.2006}.
A \textit{partitioned hashing} is described in \cite{Gupta.2017} where every hash function only gets disconnected ranges of length $m/k$ as target space. They describe that the asymptotic performance stays the same, but due to a higher number of \texttt{1}'s, they might have a higher false positive rate than the standard \ac{bf}. 
Hao et al. partition the input element set and apply \textit{different uniform hash functions} for different groups of input elements $x$ \cite{Hao.2007} optimizing for the highest fraction of zeros in the corresponding \ac{bf}. Alternatively, Bruck et al. with their Weighted \ac{bf} assign more uniformly distributed hash functions to sets of elements that have a higher probability of being queried \cite{Bruck.2006}. Other approaches proposed using \textit{non-uniformly distributed hash functions} to allow for improved functionality, especially on hardware like GPUs, where floats can be assumed instead of bits for each \ac{bf} slot \cite{Werner.2015b}. Despite additional functionality, these approaches result in increased memory footprint and complexity in comparison to our \ac{rbf}, which only relaxes existing parameter constraints. 

Additionally, \acp{bf} are proposed which adapt their size to the number of to-be-stored elements. 
The Incremental \ac{bf} \cite{Hao.2008} introduces a fill bound, which determines a lower bound to the allowable fraction of zeros and incrementally adds an additional \ac{bf} as soon as all existing \acp{bf} reach this fill bound. Similarly, it is proposed to add \textit{additional plain \acp{bf}} of increasing size $m_i=m_0\cdot s^{l-1}$ if existing filters are full and applying a geometric progression on error bounds to keep the false positive rate and the amount of elements per filter constant  \cite{Almeida.2007}. Further, they use a slicing technique that denotes different areas in the \ac{bf} for every hash function to avoid overlaps of hash values from different hash functions for one element and achieve more uniformly distributed false positive rates \cite{Chang.2004,Almeida.2007,Bose.2008}. Other approaches are the \textit{Dynamic \ac{bf}} \cite{Guo.2010}, that additionally allows for the deletion of elements from the set and is adopted from Counting \acp{bf} \cite{Fan.1998,Bonomi.2006}, and the \textit{Block \ac{bf}} \cite{Putze.2009}, which consists of a set of small cache-line-sized \acp{bf}, where each element is then only inserted in one of these subfilters. 
Furthermore, the Combinatorial \ac{bf} (COMB) and Partitioned Combinatorial \ac{bf} (PCOMB) \cite{Hao.2009} were proposed to use a set of \acp{bf} to encode situations where one element belongs to several sets. In this approach, they also partition a \ac{bf} into smaller \texttt{1}'s. However, all sub-\acp{bf} are considered to be of similar size. Last, the \textit{learned} \ac{bf}, as proposed by Mitzenmacher, is another way to reduce the size of \acp{bf} by imitating them with a learned function, which allows for false negatives and using a very small \ac{bf} to filter out these false negatives~\cite{Mitzenmacher.2018}. 
Hence, while there are approaches that split the \ac{bf} into smaller pieces, they mainly do this to allow scalability or avoid unnecessary distributed RAM access. All approaches still consider \acp{bf}, with a size being a power-of-two.. Additionally, the subparts of the filter are mostly similar in size. 

Apart from these improvements on \acp{bf} themselves Quotient \cite{Bender.2012} and Cuckoo Filters \cite{Fan.2014} were proposed, which improve, e.g., on data locality and space efficiency and allow for deletions without recalculations. With that these variants outperform many \acp{bf}. Still, \acp{bf} variants are widely used in data management systems due to their simple architecture. Our approaches allow for improvements in these domain without a general change in system architecture. 


\section{Conclusion}

We introduced the \ac{rbf} with a proven lower false positive rate than a similarly sized standard \ac{bf} due to a more optimal choice of the number of used hash functions $k$. Although the deviation of the false positive rate with respect to $k$ is close to zero next to $k^*$, it still is a more optimal solution, especially in database environments where the number of elements to be inserted, $n$, and the filter size $m$ are fixed and given by the environment. Furthermore, it is a necessity for the following propositions of more complex \ac{bf} structures, which only allow a controllable false positive rate for an optimally chosen $k^*$. Our proposal directly improves on the current state of the art for \acp{bf} and incorporates previous improvements in hashing, like double-hashing. 

With \ac{rbf} as a basis for splitting the variable length \ac{bf} into smaller pieces with a power-of-two length, we proposed \ac{vsbbf},  which allows for variable filter sizes and efficient modulo operation simultaneously. Additionally, the \ac{vsbbf} has a better distribution of false positives over all elements in the universe, as a clash of hash functions is less probable. Further, the variable size has a special impact for large \acp{bf} as they appear if the data structure is used to store large sparse set information for efficient random access. 

We relax existing constraints on \acp{bf}, which are theoretically beneficial in scenarios where either small numbers of hash functions are used or large-sized \acp{bf} are applied. The probabilistic activation of hash functions as proposed by the \acp{rbf} opens new possible research directions: It needs to be investigated how learned approaches might use rational numbers of hash functions. This extends existing approaches on learned \acp{bf}, which only use learned pre-filters to increase \ac{bf} efficiency.




\section*{Acknowledgements}
This work is funded by the Deutsche Forschungsgemeinschaft (DFG, German Research Foundation) - 507196470.

\bibliographystyle{splncs04}
\bibliography{bibliography}

\end{document}